\documentclass[]{article}

\usepackage[a4paper, total={6in, 9in}]{geometry}
\usepackage{graphicx}
\usepackage[utf8]{inputenc}
\usepackage{amssymb}
\usepackage{amsmath}
\usepackage{amsthm}
\usepackage{tikz}
\usepackage[backend=biber, style=numeric, maxnames=10, sorting=ydnt]{biblatex}
\usepackage{amsmath}
\usepackage{tabularx}
\usepackage[english]{babel}
\usepackage{float}
\usepackage{csquotes}
\usepackage[margin=10pt, tableposition=top]{caption}
\usepackage{microtype}
\usepackage{hyperref}

\usetikzlibrary{matrix}
\usetikzlibrary{cd}
\usetikzlibrary{babel}

\theoremstyle{definition}
\newtheorem{definition}{Definition}[subsection]
\newtheorem{example}{Example}[subsection]
\newtheorem{bemerkung}{Remark}[subsection]
\theoremstyle{plain}

\newtheorem{lemma}{Lemma}[subsection]

\title{Phenomenon-Signal Model: Formalisation, Graph and Application}
\author{Hans Nikolaus Beck\footnote{Robert Bosch GmbH, Abstatt}  \and Nayel Fabian Salem \footnote{Technische Universität Braunschweig, Institute of Control Engineering} \and Veronica Haber\footnote{PROSTEP AG, München} \and Matthias Rauschenbach\footnote{Fraunhofer Institute for Structural Durability and System Reliability LBF, Darmstadt} \and Jan Reich\footnote{Fraunhofer Institute for Experimental Software Engineering IESE, Kaiserslautern}}
\begin{document}

\maketitle

\begin{abstract}
Considering information as the basis of action, it may be of interest to examine the flow and acquisition of information between the actors in traffic. The central question is: Which signals does an automated driving system (which will be referred to as an automaton in the remainder of this paper) in traffic have to receive, decode or send in road traffic in order to act safely and in a manner that is compliant with valid standards. The phenomenon-signal model (PSM) is a method for structuring the problem area and for analysing and describing this very signal flow. The aim of this paper is to explain the basics, the structure and the application of this method.
\end{abstract}

\section{Motivation and Problem}\label{sec:motivation}

In the project \emph{``Verifikation und Validierung automatisierter Fahrzeuge L4/L5''} (English: Verification and validation of automated vehicles L4/L5, \emph{VVM for short}\footnote{The project \enquote{Verifikation und Validierung automatisierter Fahrzeuge L4/L5} is part of the \enquote{VDA Leitinitiative} and part of the \enquote{PEGASUS project family}, publicly funded by the German Federal Ministry for Economic Affairs and Climate Action (BMWK, http://www.bmwi.de). Web: www.vvm-projekt.de}), the fundamental question is how safety and social or legal compliance of automated driving systems can be achieved and justified methodologically. VVM builds on the results of the PEGASUS project. 
	 
From the claim that the automaton should be a member of the socio-technical traffic society, i.e. move in an urban environment (specifically: intersections) without the intervention of a human driver, the sheer number of possible events and situations impede the application of traditional engineering processes. This paper is motivated by and aims to make a contribution on how to deal with this aspect called \emph{Open Context}. 
 
As in every engineering process, the construction of such an automaton starts with the elicitation of requirements. According to the SAE \cite{sae} classification, the construction of an automaton capable of being part of a human-dominated interactive system is equivalent to the construction of a level 4/5 system. Requirements formulated as a collection of system functions are not sufficient as there is an infinite variance of possible situations an automaton can be confronted with in an open context. The ISO 26262 \cite{iso3} \cite{iso4} \cite{iso5} safety standard, which is the central safety standard for the automotive industry, is a manifestation of the idea that system functions can be thought of separately from the operational design domain and the conditions and demands of their use can be determined systematically. In the SOTIF standard \cite{sotif}, it becomes evident that industry is aware that in an open context, these conditions for using system functions can be very varied. Therefore, the system will have to realise complex behaviour patterns in which different combinations of functions are to be called up depending on the situation. Thus, the determination of requirements in an open context becomes a search for behaviour, or more precisely target behaviour. This work aims to find and describe the target behaviour together with the necessary capabilities and properties of the system.

To a large extent, traffic events are also communication events. Since the automaton itself is not a subject, but should interact with subjects in traffic, it seems reasonable to investigate information signal flows in traffic. Through appropriate modelling, it should be possible to algorithmise the formalisation presented here and implement it in a computer programme. As a basic building block for this, it will also be necessary to identify the behaviour of human road users, that is, behaviour according to legal or societal standards, which will be referred to as \textit{norm behaviour} in the following.
 
Within the framework of his philosophy, the philosopher Edmund Husserl gave detailed thought to the question of how perception and communication between subjects occur \cite{EdHrl}. An essential finding is that the intentionality of signs and words, but also the previous experience of the subjects, is constitutive for the way in which those signs and words are understood. Husserl's work strongly suggests that the flow of information alone cannot explain action. Heard or seen information becomes meaningful based on the individual history of the receiving subject; information thus becomes intentional (sender) and subjective (receiver). 

In the context of research on artificial intelligence and the resulting agent models, various parts of the acting and decision-making human have been brought into view and models have been developed. The basic question of many works is how a machine directs its actions according to a goal. This has led to various algorithms based on mathematical logic (please refer to  \cite{bradko} \cite{schalkoff} \cite{agenten}) and others).

By extending the logical means, for example with situation calculus \cite{Raymond}, approaches are being pursued that aim at logically including situations and events in decision-making in addition to goals. Another adaptation concerns modal logic, which adds the categories ``necessary'' and ``possible''. This is and has been an attempt to model the knowledge in the agent and its perception \cite{Hintikka}. All these efforts are complemented by attempts to mathematically model psychological insights into human knowledge and decision-making (cf. \cite{Tenenbaum} \cite{CompPsych}).

\section{Approach}\label{sec:ansatz}

Traffic events as communication events and traffic events as a space regulated  by legal provisions, the search for and description of target behaviour – these things call for a model that can describe both the subjective side of the individual road user with their prior knowledge and the objectification towards legal texts and the construction of an automaton.  In weighing these reasons, Husserl's ideas then appear as the most attractive approach for a model to tackle the tasks outlined above and to make a significant contribution to their solution.
  
Before presenting the approach, a general remark should be made: All considerations described here refer to a scenario of the use cases chosen in the context of the VVM project, which can be seen in Figure \ref{fig:szenariosmall}. Zones like b1, b2 etc. depicted in the figure are a suitable modelling tool for dividing areas of a scenario according to relevance. However, describing the way in which such zones can be systematically derived is beyond the scope of this paper (cf. \cite{Butz}).  

\begin{figure} [H]
	\centering
	\includegraphics[width=0.5\linewidth]{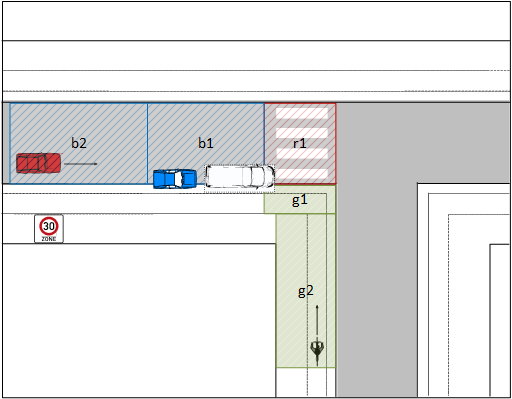}
	\caption[Abbildung]{Example scenario, ego vehicle in red }
	\label{fig:szenariosmall}
\end{figure}

Using Husserl's phenomenology \cite{EdHrl} as the basis of this work, it follows that the approach presented here should be called the \emph{phenomena-signal model}, or \emph{PSM}. Put simply, the PSM is divided into the following parts: 

\begin{enumerate}
	\item Information, knowledge and prognoses are constitutive for actions. The way in which these elements are formed based on capture by senses or sensor perception can be described by means of the ideas of phenomenology.
	\item Which actions appear possible or necessary on the basis of these elements must be described transparently.
\end{enumerate}

In the context of the PSM, these tasks are opened up and made accessible to a solution by formulating actions (which are different from human actions) as \textit{rules}. This is a constructive decision. Before proceeding further, some terms relevant to the approach need to be defined.

\section{Terms}\label{sec:begriffe}

In the definition of the presented terms, strict attention was paid to whether a statement is related to the automaton or to a human being. A human being can indeed \textit{perceive}, while in the case of the automaton, we speak exclusively of \textit{measuring} or \textit{capturing}. Likewise, an automaton \textit{acts}, while in the context of this work, we only speak of \textit{taking an action} when referring to humans. Some aspects of machine action are highlighted in \cite{Missel}, of which the concept of self-originality seems particularly interesting, which aims to distinguish pre-programmed behaviour from that of an automated system. Self-originality is the idea that a system can have more causes for interacting with the environment than pure stimulus-response or input-output mechanisms. However, the PSM is not intended as a tool for investigating these issues; its focus is rather on the capture-signal-knowledge chain. Thus, the PSM is also closer to the technology that can be realised today. 
 
The term \textit{phenomenon} is, of course, central to Husserl's work. Here, Husserl's detailed definition cannot be useful, as it is extensive and not easy to understand. Therefore, in line with \cite[436\psqq]{Brd}, \textit{phenomenon} should be understood as observable circumstances of the environment that are relevant in relation to one's intention.  If \textit{information} can be understood as a set of data related to each other, then a \textit{signal} is a piece of information of importance to the perceiving subject or the machine capturing it. Together with existing knowledge from previous experience (subject) or input knowledge (the machine), a signal contributes to the formation of predictions and action decisions. 
 
For example: For a driver who knows national road traffic regulations such as the German StVO) and has some experience seeing traffic lights turn yellow, this is a signal based on which this driver can make the prediction that they are about to lose the right of way.

As a result, Figure \ref{fig:wissenundoperator} visualises the cycle of measurement, signal, knowledge and prognosis or action. It indicates what constitutes a phenomenon. Through the interaction of these factors, capturing something in the environment becomes a conscious (subject) or processable (machine) phenomenon. The symbol $\alpha A$ depicted several times in the figure stands for an effect $\alpha$ of action A. All these elements will be explained further in Section \ref{sec:sachverhaltebasis}.
	
\begin{figure}
	\centering
	\includegraphics[width=0.85\linewidth]{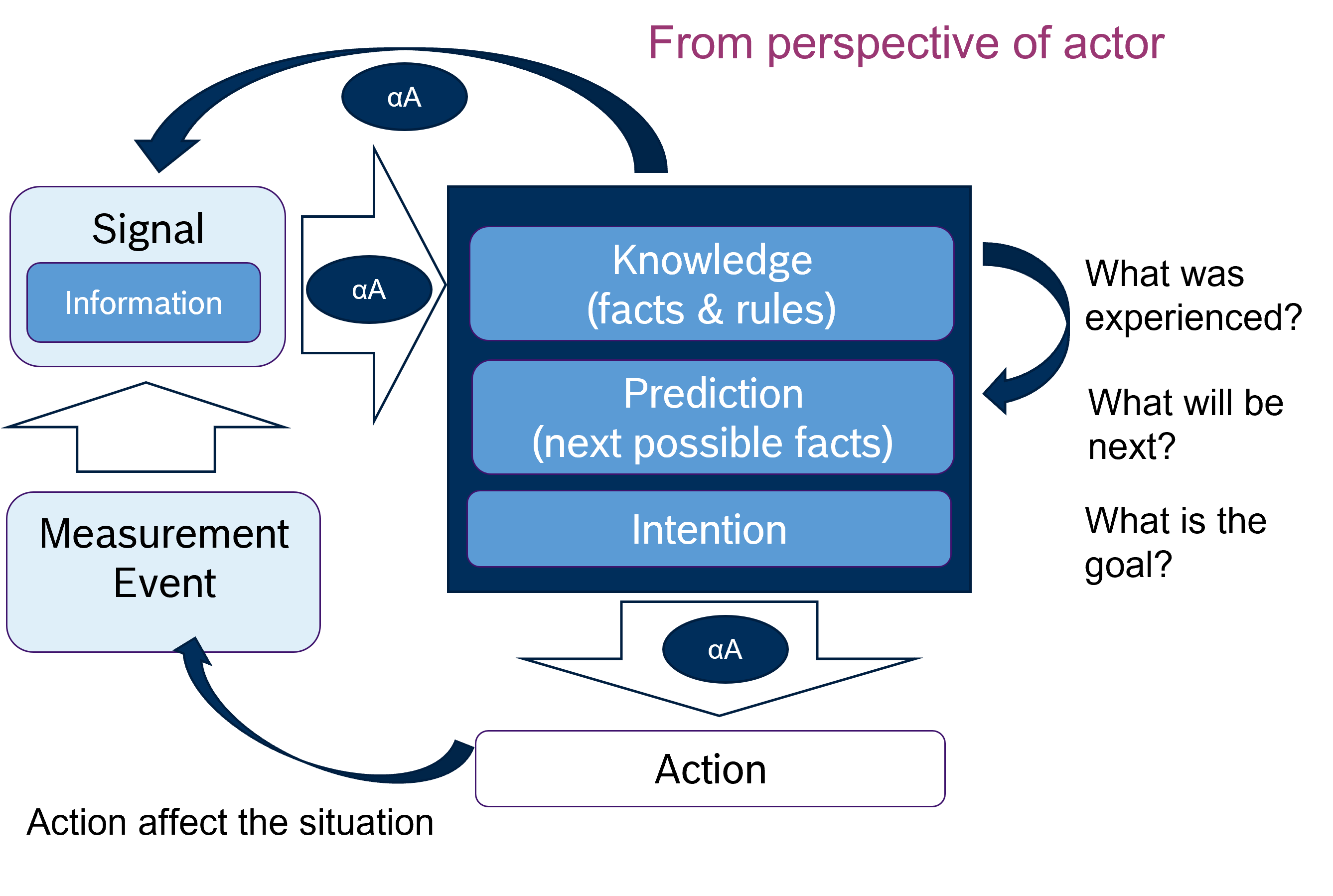}
	\caption[Ein Bild]{Measurement-signal-knowledge cycle}
	\label{fig:wissenundoperator}
\end{figure}

Rules, the essential means in this work, require clearly formulated conditions of their application. In a similar way as the legal sciences, which essentially determine behavioural norms, link legal consequences and offences in laws \cite{Eng}, \textit{facts} are all those established circumstances which justify the application of a rule, that is, represent the ``IF'' condition. Just like in the legal sciences, circumstantial evidence or factual features, here called \textit{indicators}, lead to the establishment of a fact.

``Knowledge'' in the context of the PSM thus has two aspects: On the one hand, it denotes the prior knowledge in the automaton. On the other hand, it corresponds to the captured and recognised facts. Of course, it is conceivable that these recognised facts will lead to new knowledge in the automaton, which in consequence means nothing other than learning. This aspect will be elaborated on in subsequent work. Because the rules given to the automaton are based on facts, it can be said that knowledge can be represented by facts within the framework of the PSM.

Figure \ref{fig:sachverhalte} illustrates once more the connections that are essential for this approach. Accordingly, necessary \emph{capabilities} are those which are necessary for the determination of facts and the desired actions implemented by means of rules, that is, the \emph{target behaviour}. For an exact consideration of the term \emph{capability} in the context of machines, we refer to the work of other working groups in VVM \cite{Nolte} \cite{Reschka} \cite{Bagschik}. 

\begin{figure}[h]
	\centering
	\includegraphics[width=0.8\linewidth]{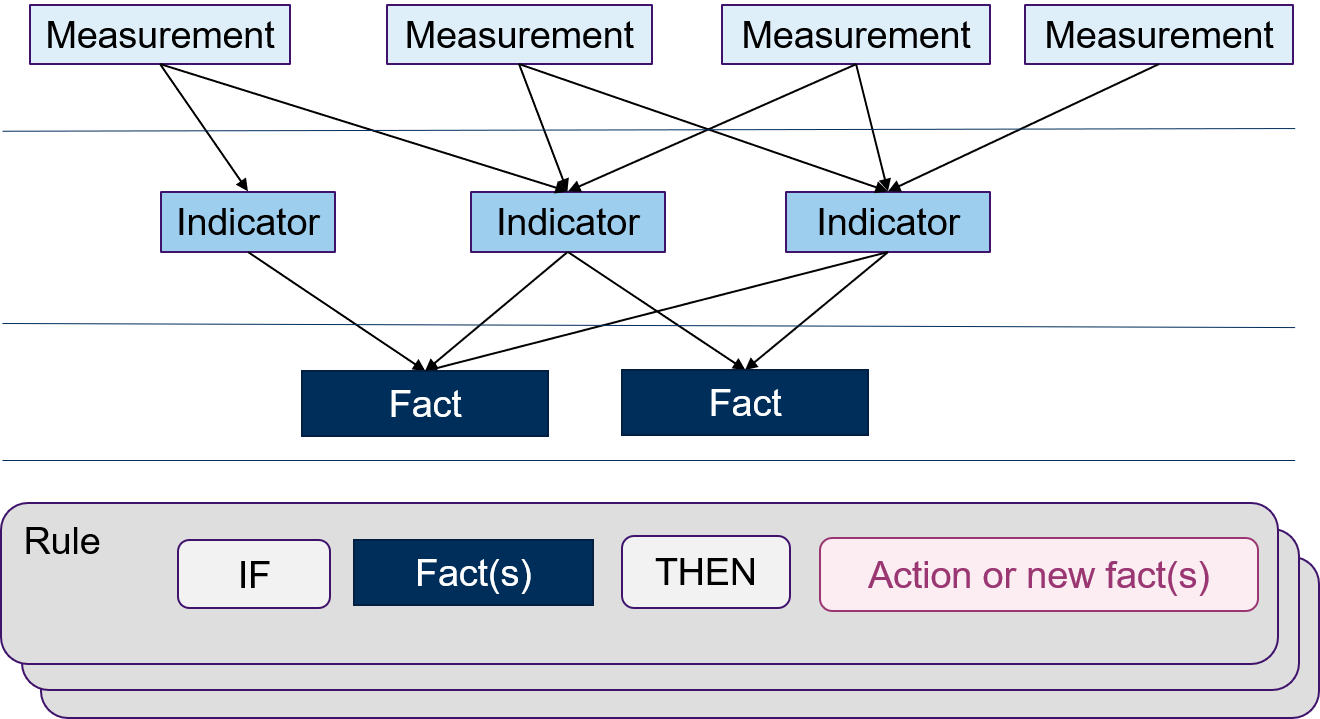}
	\caption[Abbildung]{Facts and rules}
	\label{fig:sachverhalte}
\end{figure}

Rules, and this is an essential property for the presented work, generally denote not only actions which are directed outwards, but also those which are directed inwards, which means they serve to process information or make predictions.

One may object that an action is not only based on knowledge and signals, but also essentially depends on the intention or the driving task to reach a certain goal. Since the focus of the work is on the scenario shown in Figure \ref{fig:szenariosmall}, the basic intention of the automaton is assumed to be leaving the intersection to the south or to the east in the sense of the representation. Dynamic changes of intentions are, of course, conceivable. They are connected with predictions and would have the nature of a decision. If there is a traffic jam in the city in the direction ahead, the drivers might predict that a diversion would get them to their destination faster. In the interest of a more comprehensible presentation, the formation of predictions  will not be examined in more detail here.

\section{Objectives}\label{sec:ziele}
With the descriptions up to this point, the goal of this work can now be specified: Using a suitable formalisation, it should be possible to generate a graph that can depict information flows and actions within a traffic scenario. Elements of this graph include roads, agents, various signal sources such as traffic signs or lights. This should visualise which actions are possible in relation to the available information. Structuring of the problem field and an explicit notation of assumptions are advantages that correspond to the intention of this approach, too. By applying the following three steps, this formalisation is transferred into the structure of a graph:

\begin{enumerate}
	\item Formalisation of information and its transformation into facts by means of signals.
	\item Definition of rules which are expressed using these facts.
	\item Application of all rules in order to build the graph.
\end{enumerate}

The formalisation should be done in such a way that an implementation in software is possible. The remainder of this paper is dedicated to these steps.

A note on the methodology used: This work is in the spirit of constructive-empirical testing, that is, building blocks are formed on the basis of plausible ideas. These are used to construct statements, models or solutions with the goal of recreating or essentially describing empirical situations. In general, an explanation or justification of empirical phenomena is accepted when a construction from these building blocks or a software run reproduces the phenomena with respect to the intended aspects. However, in this work, the description of the building block kit is the first goal and the basis for further work.

\section {PSM - Elements}\label{sec:sachverhaltebasis}
First of all, it should be justified why the approach visualised in Figure \ref{fig:sachverhalte} is meaningful and purposeful in the context of this paper. 

For this purpose, the question should be recalled how the capturing of aspects of the environment, information, knowledge and signals could be modelled in the context of road traffic. Human characteristics and actions in traffic, their nature in relation to information etc. must therefore be identified and applied to the machine world. An answer can obviously not be found without taking into account the specific differences between humans and machines. Neither do humans perceive objectively, nor are humans as social beings - one of the basic statements of Husserl's philosophy - independent of social manifestations, such as social norms.
 
The PSM therefore comprises two sub-considerations:  

\begin{enumerate}
	\item \emph{intentional capture} and
	\item \emph{intentional knowledge} 
\end{enumerate}

\subsection{Intentional Capture}

The term \emph{intentional capture} has to be understood as the following interpretation of Husserl's work:  A capture in the sense of one or more measurements yields information. The measurements of the size and movement of an object, for example, are connected by their relation to the object and therefore provide information about this object. But why should this object be of interest at all, or specifically its movement?

Such information is of interest to the experienced road user. It is of significance for them because their intention is to avoid collisions or assess the consequences of their actions. This intentionality in relation to information that gives meaning to the road user is what we call a \emph{signal} here. A movement from the right can be a signal for a possible right-of-way conflict. In other words, an acquired piece of information becomes a signal, that is, it becomes meaningful for the receiver, if it is an indication of something that exists in previous experience or knowledge. Fire is a sign of danger, but only if one has made harmful experiences with it. Otherwise, it could also be a signal (sign) of a heat source.

In road traffic, however, there are also deliberately placed signs, such as traffic signs, traffic lights or markings on roadways. All these things are examples of their design being intentional as well. For many people, for example, the colour red is a signal for danger due to their socialisation - a fact that warning signs and traffic lights and even brake lights make use of. So if a sender of information, in this example the traffic sign designer, is aware of the signal effect, that is, if they know which information, such as shape or colour, represents which signal for common road users, this can be exploited purposefully, that is, intentionally.
 
\emph{Intentional capture} thus means that any information in road traffic is also a signal for its receiver, the meaning of which is a function of the receiver's prior knowledge or prior experience.

\subsection{Intentional Knowledge}

The term \emph{intentional knowledge} draws on the work of the knowledge sociologists Berger and Luckmann. In their book \cite{BergLuck}, they examine how social reality is formed from the subjective experiences of individuals. For the PSM, these results are interpreted as follows: 
Social reality is conceived as social knowledge, that is, knowledge that every individual possesses. This is a deliberate simplification with regard to the objective described here. 

Due to the subjectivity explained above, each individual in society necessarily perceives their living environment - and here we limit ourselves entirely to road traffic - differently. . 
 
In the following, we will provide two examples. The way to school for the children of a school class may be dangerous at the point where the children have to cross a road. But while one parent sees the width of the road as a signal of risk, the number of parked cars may be the signal of risk for others, and yet others may be bothered by the density of traffic. From the different perceptions of the same situation, however, a common fact remains: Crossing the road at this point is risky for children.
 
In another case, the inhabitants of a city believe that the probability of an accident is comparatively high when passing a certain intersection. It does not matter whether this state of affairs (the "fact") claimed by the inhabitants may have arisen because many people actually were repeatedly involved in accidents, or whether this fate had befallen only one person who, however, was well networked in forming corresponding opinions.
 
Laws can be understood as an objectification of the sum of individual experience because their facts describe situations or circumstances that are relevant to all members of society. Hence, these facts represent the objective knowledge that results from the development of society and its phenomena. The book \cite{Eng} describes examples of how facts can also change. For example, the circumstances of theft changed because stealing electricity suddenly became possible and attractive due to industrialisation.
 
Following the construction of law, it shall thus be agreed here that \emph{facts} represent this supra-individual, i.e. objective knowledge, which is formed from the individual experiences of the individuals and is relevant for all individuals. Due to its meaning and derivation, the concept of facts stands for the \textit{intentional knowledge} of a society in the PSM. 
 
(For those who know the law, it should be noted that legal science distinguishes between matters of fact, constituent facts and facts. Not every fact is part of a matter of fact. This distinction is being dropped in the context of the PSM. The term ``matter of fact'' is not used here. As far as the PSM is concerned, only the rule-relevant, intentional knowledge referred to as ``facts'' is of importance. This will become clearer in the further sections, especially in Section \ref{sec:regeln}).

\subsection{Fact, Knowledge, Signal}
In order to make the conceptual terms developed up to this point usable, it is necessary to adopt a technical perspective. Intentional knowledge and intentional capture cannot be implemented for a machine yet. In general, one must assume that the facts of the traffic society as a social group or what the traffic regulations describe cannot be technically recorded or measured in a direct manner.

The bridge therefore is the question of how facts can be captured by a machine. Here again, legal science offers the next step. Constituent facts and circumstantial evidence are something that, when established, point to a matter of fact. This can also be observed in the physics domain, where quantities are often determined indirectly. A thermometer, for example, does not indicate the temperature directly, but actually the change in the density of matter due to the temperature. That change in density is therefore an indicator of that very change in temperature. 
 
Taking these observations into account, \emph{indicators} are to be those quantities that define the existence of a fact. The existence of an indicator can in turn be verified by one or more technically feasible measurements.

Thus, together with Figures \ref{fig:sachverhalte} and \ref{fig:wissenundoperator}, the core elements of the PSM can be described in a defining way as stated below:

\begin{definition}\label{def:psmdef}
	The following definition applies to indicators, signals and knowledge:
	
	\begin{itemize}
		\item  \emph{Facts} denote the intentional knowledge available to each participant in the social traffic society.
		\item Rules in the automaton describe actions and causality assumptions. Rules are formulated on the basis of facts (IF facts THEN...). Knowledge in the automaton is the sum of all facts and rules.
		\item 
		 \emph{Facts} also denote captured and recognised circumstances and thus the knowledge of the automaton about the environment.
		\item 
		A \emph{fact} is given when all the necessary intentional capture, mediated by \emph{indicators} defining this fact, is given.
		\item
		An \emph{indicator} is a property of the environment of the automated driving system observable by one or more technically performed measurements.  An indicator is said to be given or evident when all corresponding measurements are available.
		\item The event that an indicator has become evident corresponds to the capture event.
		\item \emph{Signal} is the element that is needed to transform a capture into a fact.
	\end{itemize}
	
\end{definition}

This closes the circle started with Sections \ref{sec:ansatz} and \ref{sec:begriffe}.

Side note: For the purpose of this work, it is sufficient to take measurements and the mechanisms behind them as given and not to question the philosophical problem of reality and positivism behind them. Entire disciplines have devoted themselves to this topic. An overview can be found, e.g., in \cite{Brd}.

\section{Symbolic Formulation}\label{sec:math}

\subsection{Sets and Maps}

The formalisation of facts uses a mathematical model that is based on sequences. The model results from the connection between facts and indicators. Facts are recognised by indicators, but the order in which these indicators are recognised may be important. This gives rise to the definitions described in this section.

\begin{definition}\label{def:basis}
	The set $C = \{ c_{1}, c_{2}, ...,c_{n}\}$ with $n\in\mathbb{N^{+}}$ of elements $c_{i}$ is called \emph{causal base set}. The elements are called \emph{causae}.
	In addition, a set $R=\{\varphi_{1}, \varphi_{2}, ...,\varphi_{n} \}$ is to be given, whose elements  $\varphi_{i}$ are called \emph{successus}. Furthermore, a mapping $\mathbf{F}$ is defined as: 
		\[
	\mathbf{F} : R \times C \longmapsto W, \quad \mathbf{F}( \varphi_{r(i)}, c_{i} ) = \varphi_{r(i)}c_{i}, \quad  \varphi_{r(i)} c_{i} \in W
	\]
	The $r(i)$ are the index functions related to $\mathbf{F}$ mapping an index of set $C$ to a index of set $R$. With this mappping, every causa of $C$ is related to a successus $\varphi$; the pairs $ \varphi_{r(i)} c_{i}$ are called  \emph{effectus}. Consequently, the set $W$ is called  \emph{set of effectus} of the causal base set. 	
\end{definition}

\begin{bemerkung}
	To simplify the notation, in the following $\varphi_{i}c_{i}$ shall always denote the corresponding pair resulting from the mapping $\mathbf{F}$. It is not intended to state that $r(i) = i$. Where it is significant, the detailed notation is used.
\end{bemerkung}
		 
\begin{example}
	The $c_{i}$ can be identified with indicators. Their effect through the successus $\varphi_{i}$ results in the possibility to measure something. For example, the indicator ``size'' allows the measurement or determination of values such as ``2 m high'', ``1 m wide''. The set $C$ then contains indicators whose effectus can be measured.
\end{example}
	
Effectus and causae can form sequences, as explained in the definition below.

\begin{definition}\label{def:seq}
	
	A \emph{sequence} $s^{k}$ of length $k$ is defined by means of the elements of an arbitrary, non-empty set $A$ as follows:
	Given $a_{i} \in  A$ and $i, k\in\mathbb{N^{+}}$:
	\begin{align*}
		s^{1} &= a_{1} \\
		s^{k} &= a_{1} a_{2}, ...,a_{k}. \\
	\end{align*}
 	The set of all sequences $S = \{s^{k1}_{1}, s^{k2}_{2},...,s^{km}_{m} \}$ with $i, k, m \in \mathbb{N^{+}}$ is called \emph{set of sequences}.
\end{definition}

\begin{definition}\label{def:map}
	Given a sequence $s^{k}$ with $k \in \mathbb{N^{+}}$ a composition is defined as:
	\[
		f : S \times S \longmapsto S, \quad 
	\]
	with 
	\begin{equation*}
		f(s^{k}_{1}, s^{1}_{2}) =  
		\begin{cases} 
			 s^{k+1} & \text{if } s^{k}_{1} \neq s^{1}_{1}  \\
			s^{k}_{1} &  \text{if }  s^{k}_{1} = s^{1}_{2}
		\end{cases}
	\end{equation*}
	
Therefore, a sequence is only extended if the new element is not the same as the last one in the sequence. To justify the first case:

	\begin{align*}
			f(s^{k}, s^{1}) = s^{k} \; s^{1} &=a_{1} a_{2}, ... a_{k} \; a_{1} \\
			&= a_{1} a_{2}, ...a_{l} \;a_{k+1}  \\
			&= s^{k+1}. 
	\end{align*}	
	
\end{definition}

\begin{lemma}
	Let $s^{k}$ and $s^{l}$ be given as two different sequences with $k,l\in \mathbb{N^{+}}$, then:
	\begin{align*}
		s^{k+l} &= s^{k} \: s^{l} \\
		&= s^{1}_{1}  \: s^{1}_{2}  \: ...  \: s^{1}_{k}  \: s^{1}_{k+1}  \: ...  \: s^{1}_{k+l}.
	\end{align*}
\end{lemma}
\begin{proof}
	According to definition  \ref{def:map} 
	\begin{align*}
		s^{k} &=   s^{k-1}  \: s^{1}_{k}.  
	\end{align*}
	$s^{k-1}$ can be expressed as
	\begin{align*}
		s^{k-1} &= s^{k-2}  \: s^{1}_{k-1}
	\end{align*}
	by applying definition \ref{def:map} and in consequence 
   	\begin{align*}
   s^{k} &=   s^{k-1}  \: s^{1}_{k} \\
   		&=   s^{k-2}  \: s^{1}_{k-1}  \: s^{1}_{k}.
   \end{align*}
   	After $k$ steps one gets
   	\begin{align*}
   	 s^{k} = s^{1}_{1}  \: s^{1}_{2}  \: ...  \: s^{1}_{k}.
   	\end{align*}
   
    If we proceed in the same way with $s^l$ and insert both representations into the assertion, then after renaming (for clarification, $s^{l}$ is written as $s^{*}$) we obtain:
    
    \begin{align*}
  	 s^{k}  \: s^{*,l} &= s_{1}  \: s_{2}  \: ...  \: s_{k}  \: s^{*}_{1}  \: s^{*}_2  \: ...  \: s^{*}_{l} \\
  	 &=  s_{1}  \: s_{2}  \: ...  \: s_{k}  \: s_{k+1}  \: ...  \: s_{k+l}\\
  	 &= s^{k+l}.
    \end{align*}

\end{proof}

\begin{bemerkung}\label{rem:interpretation}
	For the mapping $f$, different interpretations are possible:
	\begin{itemize}
		\item When  identifying the  $a_{i}$ of definition \ref{def:seq} using the elements of  $W$, then $W = \{ \varphi_{1} c_{1}, \varphi_{2} c_{2}, ...,\varphi_{k} c_{k} \} = \{s^{1}_{1}, s^{1}_{2}, ..., s^{1}_{k}\}$ with $s^{1}_{i} \in S_{W} \supseteq W$ and  $k\in\mathbb{N^{+}}$ (once more, the lower index serves for differentiation only).  
		
		\item Setting  $a_{i} = c_{i} \in S_{C} \supseteq C$ and $f_{C} : S_{C} \times S_{C} \longmapsto S_{C} $ means the mapping $f_{C}$ generates sequences like $s^{k} = c_{1}c_{2}...c_{k}$.
		
		\item When $a_{i}$ is set equal to the elements of $R$, then analogously to the procedure above, $f_{R} is: S_{R} \times S_{R} \longmapsto S_{R}$ and  $s^{k} = \varphi_{1}\varphi_{2}...\varphi_{k} $.
		
	\end{itemize}
\end{bemerkung}

\begin{bemerkung}
	
	To simplify the notation, the superscript index for the length specification can also be omitted, that is, instead of writing $s^{k}_{i}$ for a sequence of length $k$, only $s_{i}$ is written. With this designation, a length $k \in \mathbb{N^{+}}$ is implied in this case, but not indicated explicitly.
		
\end{bemerkung}

\begin{definition}
The mapping $\mathbf{H}$ is defined as follows: 
\[
\mathbf{H} : S_{R} \times S_{C} \longmapsto \overline{W}, \quad \mathbf{H}( \varphi_{h(i)}, c_{i} ) = \varphi_{h(i)}c_{i}, \quad  \varphi_{h(i)} c_{i} \in \overline{W}
\]
Again, $h(i)$ are the index functions belonging to $\mathbf{H}$ for the assignment of an index of the set $S_{C}$ to an index of the set $S_{R}$. With this mapping, each sequence from $S_{C}$, which here has the meaning of a causa, is assigned a sequence from $S_{R}$, which represents a successus.  Accordingly, the set $\overline{W}$ is called \emph{effect set} of the causal basic set. 

 \par 

\end{definition}

All the mappings introduced so far can be visualised in the following \emph{sequence realisation diagram}: 

\begin{figure}[h]
\centering

\begin{tikzcd}[row sep=huge, column sep = huge]	
	R \times C 
		\ar[r, "f_{C}; f_{R}"]
		\ar[d, "\mathbf{F}"]
	 &  S_{R} \times S_{C} \arrow[d, "\mathbf{H}", shift left=2ex] \\
	W 	\arrow[r, "f"]		& S_{W}, \overline{W} \ar[lu, dashrightarrow, shift left = 0.7ex] \ar[lu, dashrightarrow, shift right = 0.7ex]
\end{tikzcd}	

\caption{Sequence realisation diagram}
\label{fig:commute}
\end{figure}
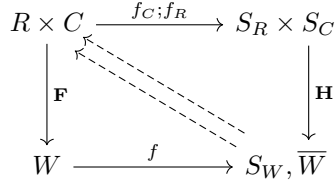

If $S_{W} = \overline{W}$ were to hold and thus commutate the diagram, this would mean that the sequence composition would result in the same set of sequences regardless of the path. The following definition describes a criterion for this.

\begin{definition}
	
The sequence realisation diagram according to Figure \ref{fig:commute} commutes, that is, $S_{W} = \overline{W}$ holds if an isomorphism $\mathbf{T} : S_{W} \longrightarrow \overline{W} $ exists with

\[
	 s_{1} = \varphi_{1}c_{1}\varphi_{2}c_{2},...,\varphi_{k}c_{k} \in S_{W}
\]	
and as a result of $\mathbf{H}$
\[
		\overline{s_{2}} = \varphi_{1}\varphi_{2},...,\varphi_{k}c_{1} c_{2},...,c_{k} \in \overline{W},
\]
it shall apply: 	
\[
 \mathbf{T} : S_{W} \longrightarrow \overline{W}, \{s_{1} \in S_{W}, \overline{s_{2}} \in \overline{W} \;|\; s_{1}  = \overline{s_{2}}  \}.
\]

\end{definition}

Because definition \ref{def:basis} does not impose requirements for the elements of set $C$, 
 all elements of set  $\overline{W}$ can be taken as elements of a set $C^{''}$ representing a 2nd order causal base set. Therefore, $\overline{W} \subseteq C^{''}$. Together with an arbitrary set $R^{''} = \{\omega_{1},\omega_{2}, ..., \omega_{o} \}$ of second-order successus and by applying $\mathbf{F}$, we get sequences of the form
  
\begin{align}
	\omega_{j} \; \varphi_{1} c_{1} \varphi_{2} c_{2} ... \varphi_{p}c_{p} =&  
	\omega_{j} \; \varphi_{1} \varphi_{2} ... \varphi_{p} c_{1} c_{2}... c_{p} \label{eq:commute}
\end{align}

with $p,o \in \mathbb{N^{+}}$ and $j \in \{1, 2, ... o\}$. 

\begin{bemerkung} ~
	
	\begin{itemize} 
		\item It is easy to see that Figure \ref{fig:commute} can assume any order in all elements.
		\item Equation \eqref{eq:commute} now makes it possible to talk about successus of sequences. Since - as stated in remark \ref{rem:interpretation} - the $\varphi_{i} c_{i} = s^{1}$, that is, the first-order effectus, can also be understood as sequences of length 1, further definitions and propositions are only expressed via sequences.
		
	\end{itemize}
\end{bemerkung}

\subsection{Calculus}\label{sec:kalkuel}

From the previous section, it became clear that any sequences $s^{k}_{i}$ can be assigned to successus and the results can be concatenated to higher-order sequences like $s^{'} = \omega_{1} s_{1} \omega_{2} s_{2} ... \omega_{n} s_{n}.$ The definitions that follow below characterise special properties of the successus, such as neutrality or constancy, and prepare the basis of a calculus. First, however, it must be specified what is meant by ``equality''.

\begin{definition}\label{def:gleich}
	
Let $s_{i}, s_{k}$ be arbitrary sequences with $ s_{i}, s_{k} \in \overline{W}$. The sequences are called \emph{E-equal} if $\omega \in R^{(n+1)} $ holds true while taking an arbitrary successus of the order $n+1$.
\[
	s_{i} \stackrel{E}{=} s_{k} \Longleftrightarrow \omega\:s_{i} = \omega\:s_{k}
\]
	
\end{definition}

Using these preconditions, special sequences can now be specified.

\begin{definition}
	A successus $\varphi^{c} \in R^{(n+1)}$ is called \emph{constancy} if:
	\[
	\varphi^{c} s_{i} = \varphi^{c} s_{j} \quad \text{ for all} \quad s_{i}, s_{j} \in \overline{W}.
	\]
\end{definition}

Constancy thus prevents the causae given by the sequence from having an effect. 

\begin{definition}
	A successus of n-th order $ \mathcal{I}$ is called \emph{neutrum} if
	\[
	\mathcal{I} s_{i} \stackrel{E}{=} s_{i} \quad \text{ for all} \quad s_{i} \in \overline{W}.
	\]
\end{definition}

\begin{definition}
	A successus of n-th order $\omega^{-1} c$ is called \emph{inverse effect} or \emph{neutraliser} if
	\begin{align*}
		 \:\omega s_{i} \omega^{-1} s_{j} & \stackrel{E}{=}  \omega \omega^{-1} s_{i} s_{j}  \\ & \stackrel{E}{=} 
		  \:\mathcal{I}s_{i} s_{j} \quad \text{	for all} \quad  s_{i}, s_{j} \in \overline{W}.
	\end{align*}
\end{definition}

\begin{lemma}\label{lem:nodouble1}
	When $\omega s_{1} \omega s_{2} $, then it holds that
	
	\[
	 \omega s_{1} \omega s_{2}  \stackrel{E}{=} \omega s_{1}s_{2}.
	\]
	
\begin{proof}
	Because the diagram \ref{fig:commute} commutes, the sequence $\omega s_{1} \omega s_{2} $ can be expressed with the mapping $f$ of \ref{def:basis} as 
	\[
		f(\omega, s_{1}) f(\omega, s_{2}) \stackrel{E}{=} f(\omega, \omega) f(s_{1}, s_{2}) = \omega s_{1} s_{2}
	\] 
in consequence of the property of $f$ described in case 2.
	
\end{proof}

\end{lemma}\label{lem:nodouble2}
So there is no duplication of the same effects within a sequence. With the same structure of proof it follows

\begin{lemma}
	For a given $\omega_{1} s \omega_{2} s$, it holds that 
	
	\[
	\omega_{1} s \; \omega_{2} s  \stackrel{E}{=} \omega_{1} \omega_{2} \;s.
	\]
\end{lemma}

Under certain conditions, the order of the causae can be swapped. This is described by the following lemma.

\begin{lemma}
	If $\omega s_{i}$ and $\omega s_{j}$ are given as two effectus of $\overline{W}$ and it holds that
	\[
		\omega s_{i} \omega s_{j} \stackrel{E}{=} \omega s_{j} \omega s_{i} 
	\] then it holds that
	\[
			 s_{i} s_{j}  \stackrel{E}{=}  s_{j} s_{i}
	\]
	
\begin{proof}

	\begin{align*}
 \omega s_{i} \omega s_{j} &\stackrel{E}{=} \omega s_{j} \omega s_{i} \\
	\omega\:\omega s_{i} s_{j}  &\stackrel{E}{=}	\omega\:\omega s_{j} s_{i} \\
	\omega s_{i} s_{j}  &\stackrel{E}{=}	\omega s_{j} s_{i} \\
 	s_{i} s_{j} &\stackrel{E}{=} s_{j} s_{i}
	\end{align*}

\end{proof}
\end{lemma}

\section{PSM Symbolism of Facts}\label{sec:symbolik}

Now the question arises which measurements and which indicators result in facts. Measurements are strongly influenced by technical constraints. Working groups in the VVM are dealing with corresponding sensor models. In order to keep the discussion comprehensible, indicators are only defined abstractly.  However, one thing should be noted about the nature of this relationship: In general, the question of which indicators make a fact appear evident depends on the presumptions of causality that science and experience provide us with.
An empirical science hypothesises relationships on the basis of observations, which can be exploited to find indicators. In the case of the thermometer mentioned above, the temperature-dependent expansion of a liquid is such a causal assumption and thus a possible option for an indicator for temperature determination.

Normative sciences such as legal science define such causalities, or they emerge as social consensus \cite{BergLuck}.  ``Theft is when someone takes property from another'' would be such a set causality, greatly simplified here for illustration purposes.  All that matters is that the relations between facts, their indicators and measurements are a model representation of what is considered significant or intentional in the world of society. The choice of indicators and facts is thus dependent on the application on the one hand, but on the other hand, this also enables the statements of the PSM to be finely tuned to the question. As far as the present work is concerned, indicators are used as examples and without deeper justification. For details on the choice of facts and indicators, see \cite{Salem}.

Table \ref{tab:indikatoren} shows an example selection of indicators, their symbols and semantics. This selection has been made in view of Figure \ref{fig:szenario} and with the intention of presenting the basic ideas in an understandable way so that it can be easily adapted for extensive use cases.

\begin{table}[h]
	\caption[Tabelle]{Indicators}
	\small
	\label{tab:indikatoren}
\begin{tabularx}{\textwidth}{|c|c|l|X|}
	\hline 
	\textbf{Symbol} & \textbf{Indicator} & \textbf{Domain} & \textbf{Explanation} \\ 
	 &  Causa &  Successus &\\ 
	\hline 
	P & Position & b1, b2, g1, r1 & Reference to a position of a capture. This reference is made using zones according to Figure \ref{fig:szenario} \\ 
	\hline 
	A& Extension & r, f, a, l &  Extension of a capture, i.e., extension of a pedestrian f, cyclist r, car a or truck l \\ 
	\hline 
	Q& Quality & ü, r, f, a, l &  Like A, but as an abstract quality, for example reflection properties etc. ``ü'' is the symbol for a pedestrian crossing\\
	\hline
	R& Direction & $<$, $>$, +, -& Direction of the capture relative to the ego ,"+" forward,  ``$<$'' left, ``$>$'' right  \\ 
	\hline 
	B& Movement & 0, $<$, $>$, +, -& Movement according to the respective direction. 0 means no movement or stop  \\ 
	\hline 
\end{tabularx} 

\end{table}
\paragraph{}

These indicators are now embedded into the symbolic formalism according to Section \ref{sec:math} by forming the set $C$ from definition \ref{def:basis}. Their value range is the value range of the successus $R$ according to this very definition. This results in:

\begin{align*}
C &= \{P, A, Q, R, B\} \\
R &= \{b1, b2, g1, r1, r, f, a, l, <, >, + , -, 0\}
\end{align*}

Writing Table \ref{tab:indikatoren} in mathematical notation would result in: $b1P \in W, rA \in W, rQ \in W$, but, for example, $<P = \mathcal{I}$. All pairs $\varphi_{r(i)} c_{i}$ that are not covered by the range of values according to the table above realise only the neutrum $ \mathcal{I}$. By means of this interpretation, Table \ref{tab:indikatoren} can be completely defined as a set $W$. 

With reference to the sequence realisation diagram \ref{fig:commute}, effectus sequences can be built up by concatenation with $f$. The transition to the 2nd order, that is, 2nd order successus of these effectus sequences, describes captures and facts by means of the following definition: 

\begin{definition}
According to the description provided above, i.e.,
\[
C^{''} := W = \{b1P, b2P, ..., rA...rQ..., +R..,\mathcal{I}\}
\]

let the 2nd-order causae be (= $\varphi_{f(i)} s_{i}$) and
\[
R^{''} = \{!, ?, !^{-1}, ?^{-1}\}
\] 

the 2nd-order successus, including their inverse effects. Then the elements of the mapping $\mathbf{F}(\varphi,s) = \:?\:\varphi_{f(i)}s_{i} $ are called \emph{capture events}, and the elements of the mapping  $\mathbf{F}(\varphi,s) = \:!\:\varphi_{f(i)}s_{i} $ are called the \emph{facts} of a scenario.
\end{definition}

It is clear that the set of all effectus sequences $\varphi_{f(i)} s_{i}$, which has the nature of sequences of realised indicators, determines all facts and captures in the model that can be represented at all. Facts thus represent themselves as a 2nd-order successus with a causa, formally $!\:\varphi_{f(i)}s_{i} $ with $i=1...k$ and $k \in \mathbb{N^{+}}$, where the causa consists of an effectus sequence $\varphi_{f(i)} s_{i}$ of length $k$. 

By choosing the indicators and their realisations (the successus), the model can thus be completely adjusted to the problem. Now it is possible to formally represent knowledge, facts and rules in a machine by applying the calculus according to Section \ref{sec:kalkuel}. 

\section{PSM -  Rules}\label{sec:regeln}

With the facts encoded in a symbol language, rules can now be written down compactly and simply according to Section \ref{sec:ziele}.  Rules describe different aspects of the overall model. To clarify these, we once again refer to Figure \ref{fig:wissenundoperator}. All processes shown there with the symbol $\alpha A$ are subject to rules. According to their role, the following rule types can be distinguished:

\begin{figure}[h]
	\centering
	\includegraphics[width=0.7\linewidth]{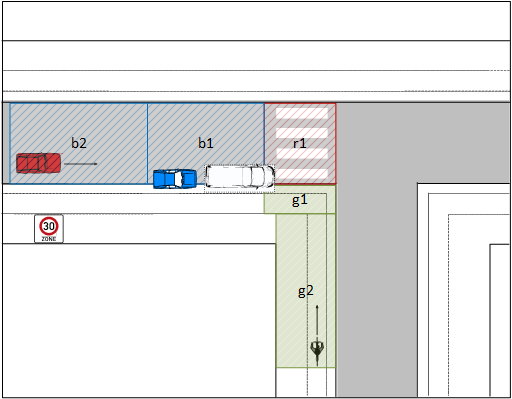}
	\caption[Abbildung]{Example scenario }
	\label{fig:szenario}
\end{figure}

\begin{itemize}
	\item \emph{Behavioural rules}: As illustrated in Figure \ref{fig:wissenundoperator}, knowledge and prognoses trigger actions. In the model represented here, facts represent knowledge. Rules that describe the action that follows from a fact and is visible to the outside are therefore called behavioural rules. 
	
	On the one hand, behavioural rules model  \emph{norm behaviour}, that is, expectations regarding behaviour that are based on the facts formed in social life (e.g. laws). On the other hand, they also represent the \emph{target behaviour}, which is the behaviour that the vendor of the automated vehicle would like to have as a product decision or that is input to the development as a requirement from laws or norm behaviour.  The semantic analysis of PSM \cite{Salem} is a method for identifying such behavioural rules.
		 
	\item \emph{Structural rules}: The $\alpha A$ in Figure \ref{fig:wissenundoperator} also represent the recognition of a signal as such. As already described in detail, this process depends on the existing knowledge, given in the form of rules and facts. According to definition \ref{def:psmdef}, describing a signal means noting or formalising the path from capture to fact. 
			 
	Modelling the path from captures to indicators is easy because physical laws will mostly play a role here. Recall the example of the thermometer mentioned above and also Figure \ref{fig:sachverhalte}. Structural rules at this point of the model can thus be obtained from the causal assumptions mentioned in Section \ref{sec:symbolik}.
				 
	Less obvious is the path from indicators to recognised facts. In the Husserlian sense, signals convey indications that -- to move on to the language of PSM -- stand for one or more rules already present in the machine and their facts. 			 
	
	Recall that a fact is a sequence (of effectus). If all indicators are realised, then the fact defined by them is captured, that is,  $?\:\varphi_{f(i)}s_{i} $. If this sequence $\varphi_{f(i)}s_{i}$ is found as part of a deposited fact of the form $?^{-1}\:\varphi_{f(i)}s_{i}$, then according to the formalism, $?\:\varphi_{f(i)}s_{i} $ becomes $!\:\varphi_{f(i)}s_{i} $ and the fact is recognised.
	
	It is precisely this linkage that is also expressed by means of structural rules and, because it completes the path from captures to facts, it does, together with the expressions of the type $?^{-1}\:\varphi_{f(i)}s_{i}$, correspond to the representation of a signal in the PSM.
		
	Structural rules also describe structural developments of a scene, such as visibility conditions (who can see whom and from where) or movements of objects or agents.
	
	\item \emph{Equivalence rules}: This kind of rules prescribes transformations of (indi\-cator-) sequences which do not change the overall effect. In the language of the formalism, this corresponds to the concept of ``E-equal'' according to definition \ref{def:gleich}. An example will illustrate this later. Rules of this class are useful, for example, when a measurement can be replaced by a technically more accessible measurement.

\end{itemize}

In order to formalise the given rules specifically, we refer again to the underlying example scenario. It should be noted at the outset that

\begin{itemize}
	\item the aim here is not to achieve geometric accuracy in the rule statements, but rather to illustrate the method and the model structure,
	\item the rules, especially the behavioural rules, are always to be formulated from the point of view of the ego vehicle; other agents are only taken into account via structural rules, 
	\item  uppercase letters X, Y, Z in the symbolic sequences stand for any indicators,
	\item lowercase letters x, y, z represent their successus,
	\item the specified rules represent applications of the calculus. Several conditions of a rule are linked with logical AND.
\end{itemize}

\begin{table}[h]
	\caption{Structural rules related to dynamics}
	\small
	\begin{tabular}{p{0.15\textwidth}m{0.3\textwidth}m{0.5\textwidth}}
		\hline
		Issue & Rule & Meaning\\
		\hline
	 	Visibility&
	 		\begin{tikzcd}[row sep=0]
	 			Xb1P \ar[rd]   & \\
	 		 & ?Yr1P\\
	 		Yr1P \ar[ru]	 &
	 	\end{tikzcd} & From zone b1, something can be captured in r1, provided something is there \\
 	\hline
		Visibility & 	\begin{tikzcd}[row sep=0.1]
			Xb2P \ar[rd]    & \\
			 & ?Yr1P\\
			Yr1P \ar[ru] &
		\end{tikzcd}  & Same for zone b2\\
	\hline
		Visibility & \begin{tikzcd}[row sep=0]
			Xb2P \ar[rd]    & \\
		 & ?Yg2P\\
			Yg2P \ar[ru] &
		\end{tikzcd}  & Something in zone g2 can be captured from zone b2.\\
	\hline
		Movement & +Bb2P $\longrightarrow$ +Bb1P & When an agent moves, their position changes from b2 to b1 and the movement is maintained\\
		\hline
		Movement & +Bb1P $\longrightarrow$ +Br1P & Same for b1 to r1 \\
		\hline
		Movement & X+Bg2P $\longrightarrow$ X+Bg1P & An object in g2 moves to g1 (in this example a cyclist).\\
		\hline
		Movement & +Bg1P $\longrightarrow$ +Br1P & g1 to r1\\
		\hline
		\\
		\end{tabular} 
	\label{tab:structure1}

\end{table}

Figure \ref{fig:szenario} shows that the ego vehicle (red) can only ``see'' the cyclist when it is in zone b2 and the cyclist is in zone g2. Furthermore, the ego vehicle also sees zone r1, that is, the pedstrian crossing, from b2 and b1. Statements of this kind have to be expressed in structural rules according to Table \ref{tab:structure1}.

\begin{table}[H]
	\caption{Structural rules related to facts}
	\small
	\begin{tabular}{p{0.15\textwidth}m{0.3\textwidth}m{0.5\textwidth}}
		\hline
		Issue & Rule & Meaning\\
		\hline
		Signal  & 
		\begin{tikzcd}[row sep=0]
			?s_{i} \ar[rd] &    \\
			& !s_{i}\\
			?^{-1}s_{i}!s_{i} \ar[ru]	 &
		\end{tikzcd} & because of Lemma \ref{lem:nodouble1} and \ref{lem:nodouble2},
		{\begin{align*} 
			&	?s_{i} \; ?^{-1} s_{i} \; ! \; s_{i}  \\
			&\stackrel{E}{=} \;	? \; ?^{-1} \; !\; s_{i} s_{i} s_{i} \\ &\stackrel{E}{=} \mathcal{I} \;! \;s_{i} & 
			\end{align*} }
		A capture becomes a fact by applying signal $	?^{-1}s_{i}!s_{i}$.  \\
		\hline
	\end{tabular} 
	\label{tab:structure2}
\end{table}

The totality of these rules does not provide an option for the cyclist to be noticed by the ego vehicle when staying in zone g1. A zone that cannot be seen from a certain zone is an example of a criticality phenomenon of occlusion. This will not be explored further here, but it shows that accident-related phenomena or findings from hazard analyses can also be represented via the modelling of causality assumptions via rules.  For further phenomena of this kind and more details, see also work from sub-project 2 of the VVM \cite{Neu}.

\begin{table}[h]
	\small
	\caption{Behavioural rules}
	\begin{tabular}{p{0.15\textwidth}m{0.3\textwidth}m{0.5\textwidth}}
		\hline
		Issue & Rule & Meaning \\ \hline		
		StVO Rule & 	
		\begin{tikzcd}[row sep=0, column sep=small]
			!+B	\ddot{u}Q+R \ar[rd] &   & \\
			& 0B	\\
			!rQg1P \ar[ru]	 &
		\end{tikzcd} 
		& These circumstances are given: pedestrian crossing ahead, ego in front of pedestrian crossing and driving, cyclist at pedestrian crossing, therefore stop.\\ \hline
		StVO Rule & 	
		\begin{tikzcd}[row sep=0, column sep=small]
			!+Bb1P	\ddot{u}Qr1P \ar[rd] &   & \\
			& 0B	\\
			!rQr1P \ar[ru]	 &
		\end{tikzcd} 
		& Same fact as above, alternative representation; cf. alternative rules. Not used in the example graph.\\ \hline
		Collision &
		\begin{tikzcd}[row sep=0, column sep=small]
			XxP \ar[rd] &   & \\
			& 00	\\
			YxP \ar[ru]	 &
		\end{tikzcd}  &Things in the same zone collide (equal successus of P)  \\ \hline
	\end{tabular} 
\label{tab:verhalten}
\end{table}
 
Table \ref{tab:structure2} contains the rules concerning signals. Behavioural rules for the given example are taken from the German Road Traffic Regulations and somewhat simplified. Essentially, the ego vehicle should stop when it detects the cyclist at or on the pedestrian crossing. The following points should be taken into account when formulating the rules:

\begin{itemize}
	\item Behaviour means an action is necessary to maintain the existing state or an existing equilibrium between automaton and environment. Driving on an open road would be an example of such equilibrium. Although this requires stepping on the gas, or pedalling on the bicycle, it is not an action in the PSM! An action is necessary if it is the only way to support the intention of the ego vehicle, such as reaching the destination. 
	\item Behavioural rules are behavioural rules of the ego vehicle. Agents are covered by structural rules.
\end{itemize}

An equivalence rule can also be formulated for this example, which replaces naming of the next zone as ``zone ahead''. These rules are listed in Table \ref{tab:equivalent}.

\begin{table}[h]
	\caption{Equivalence rules}
	\small
	\begin{tabular}{p{0.15\textwidth}m{0.3\textwidth}m{0.5\textwidth}}
		\hline
		Issue & Rule & Meaning \\ \hline		
		Forward& !b1PxQr1P $\longrightarrow$ !xQ+R & If ego vehicle is in b1 and something is given in r1, it is equivalent to the fact that an object xQ is ahead in the direction of travel.
	\\ \hline
	\end{tabular} 
\label{tab:equivalent}
\end{table}

\section{PSM Graph}
Using the rules presented in the previous section, a graph for Figure \ref{fig:szenario} can now be created. Recall that the role of information for acting is to be investigated. This is a question of the structure of information and its flows. With this claim in mind and with the formalisms worked out here, the graph can now be created according to the following algorithm, which is, in principle, suitable for machines:

\begin{enumerate}

\item Place the structural conditions for each agent as a starting node into the graph. In the example, these are $\ddot uQr1P = $ pedestrian crossing in r1, $+Bb2P = $ ego vehicle in b2 and driving, $rQg2P = $ cyclist in g2.
\item Create the next node by applying all applicable rules. The rule sequence defines the node that is newly created. In the example, a structural rule of $\ddot uQr1P$ and $+Bb2P$ would result in the node $?Yr1P$ with $Y=\ddot uQ$ and so on.

\item Repeat step 2 until none of the available rules are applicable anymore.
\item If necessary, delete all paths which do not lead to an action node.
	
\end{enumerate}

The colours have been added manually for illustration purposes. Grey represents structural circumstances. Light blue indicates nodes of capture events, green those of facts. Orange corresponds to signals whose signature needs to be present in the prior knowledge. It should be noted that rules with multiple conditions are only to be applied if all conditions are fulfilled (AND relationship).

\begin{figure} [H]
	\centering
	\includegraphics[width=1.0\linewidth,height=0.60\textheight]{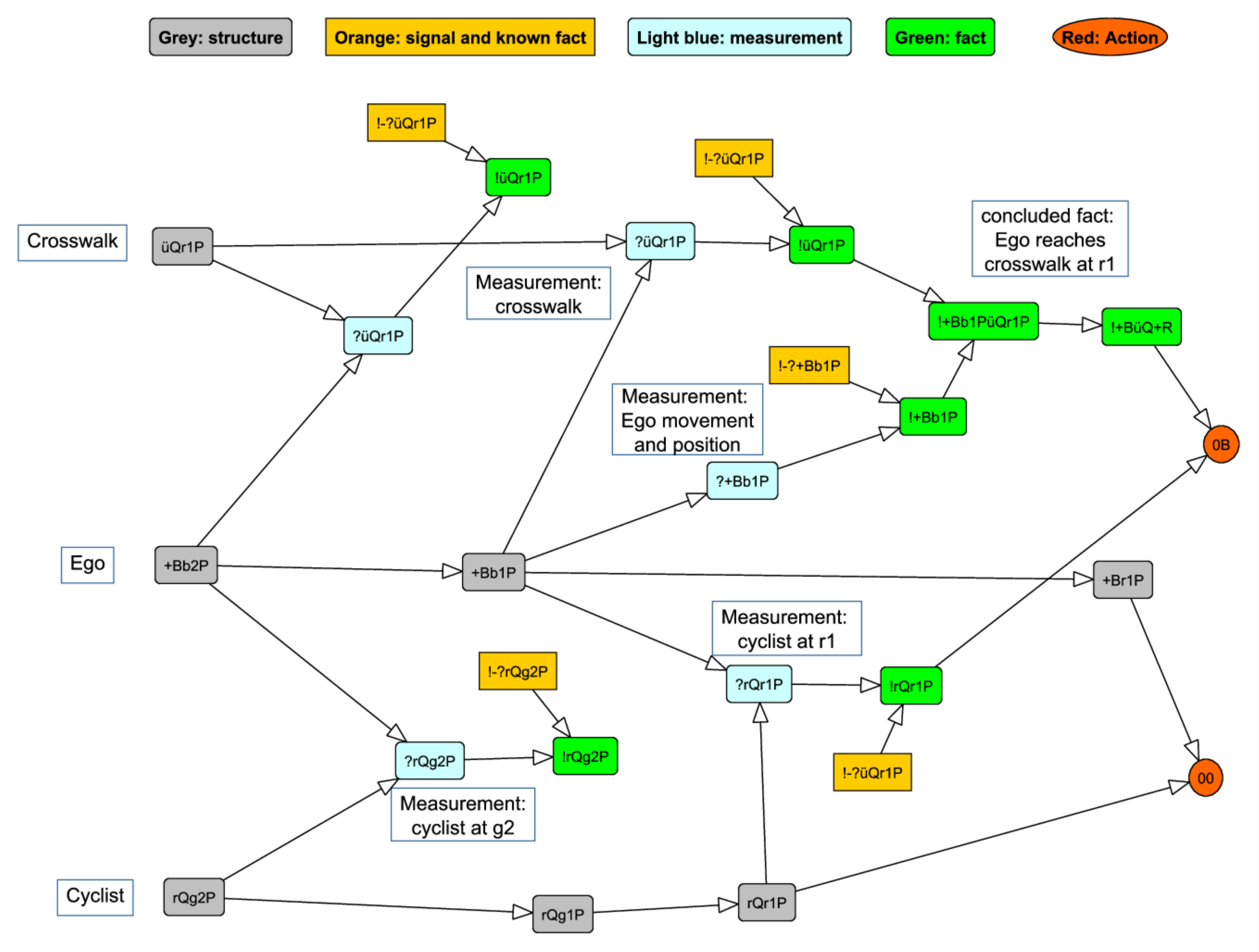}
	\caption[Abbildung]{Example PSM graph}
	\label{fig:psmgraph}
\end{figure}

As a result, this graph shows different paths from the start condition of the scenario to the actions, where different paths represent different information and signal flows. Different evaluations can now be carried out on the PSM graph by looking at paths through the graph. These paths have no direct correlation to time, but they are nevertheless sequences and thus represent a progression. In other words, the graph symbolises the evolution of informational structures of the inherently temporal scenario.

The example graph depicted in Figure \ref{fig:psmgraph} shows that there are paths from the grey starting nodes to the collision node (00). And as can be expected from experience, non-capture leads to an undesired event: both paths exclusively run via grey nodes, no node of the capture is traversed. Conversely, all paths leading to braking of the ego vehicle (0B) run via capture nodes and facts. 

As an example, consider the path for the detection of the cyclist. For the ego vehicle, this would be [+Bb2P, +Bb1P], for the cyclist [rQg2P, rQg1P, rQ1P]. If the ego vehicle is in b1 and the cyclist is in r1, that is, on the pedestrian crossing, the ego vehicle can detect him and initiate braking by forming the fact ``cyclist on pedestrian crossing '' (green node).

Although only one small example has been addressed and a small extract has been presented, some statements can be derived here:
\begin{itemize}
	\item Braking requires a fact. In order for the ego vehicle to form this fact, it needs to  understand the signal (the orange node represents the necessary prior knowledge for this).  
	\item From this it follows that the car needs the \textit{capability} to perform, among other things, the capture ?rQr1P (see path) and form the fact !rQr1P. 
	\item If a collision is to be avoided, further rules (or modified ones) are needed so that the node (00) can no longer be reached.
	\item Should the ability to form the facts be skipped (because of costs) and the ego vehicle be supplied with braking independent of indicators, rules would have to be expressed that lead directly from capture to braking (0B).
	\item A \textit{target behaviour} is represented by exactly those paths in the graph of the ego vehicle that are to be realised, described by the related behavioural rules and capabilities. The latter are indicated by the capture and fact nodes (incl. the signals) in this path.
\end{itemize}

This is where the nature of the PSM graph and the formalism for its creation as a tool become apparent. This graph structures the problem area and makes it accessible for evaluation and closer, targeted examination. In addition, the formalisation of indicators and facts allows the transparent justification of the relevance and meaningfulness of the modelled elements and assumptions and thus supports a safety-related argumentation. This follows from the described path of deriving the model constructs (indicator, fact, etc.) all the way to the graph. The decisions along this derivation are always justifiable.

\section{Summary and Future Work}

With the PSM concept, the decision-making possibilities of an L4/L5-automated vehicle are represented depending on situationally received signals about the surrounding traffic space and the behaviour of other road users. Changes in the traffic situation are inferred from the respective rule-based decision-making possibilities. In addition, the information subsequently available to the vehicle and the decision options derived from this information are checked and the graph is developed further for these. This is achieved through the formalisation of behaviour, knowledge and signals. The concept reveals possibilities for the representation and evaluation of the vehicle's behaviour and resulting scenario progressions. This shall be used in subsequent work for the development of the desired vehicle behaviour as well as its validation and verification.

In the further course of the project VVM, as mentioned above, the computer-assisted generation and evaluation of the PSM graphs shall be pursued. This is necessary for a practical application of the concept in order to be able to cope with the large number of possible progressions and variations of the scenario in reality. Based on such a possibility, accompanying questions for using the concept for system definition and verification of the sufficient safety of the implemented vehicle behaviour in traffic will be investigated.

Among other questions, the methodologically coherent and consistent generation of PSM graphs in interaction with scenarios and criticality phenomena from hazard analyses plays a role. In this context, a systematic derivation of requirements and specifications for a safety-oriented system architecture with regard to the required system behaviour is to be achieved. Another question is how the PSM concept can be implemented in order to be able to present a target-oriented definition of the totality of a vehicle's behaviour that is appropriate in each situation.

Finally, a systematic approach and an evaluation procedure are needed to determine the extent to which the capabilities designed to date, as well as the technically specified functions, ensure a vehicle's sufficiently risk-minimised operation in traffic. These should also provide evidence of the improvements or safeguards that can sufficiently mitigate the inadequacies. This is ultimately the basis for the verifiability of a vehicle's behaviour as well as its testability and verification. 

Development of the model itself is also continuing. The concatenation of rules and rule hierarchies, in particular, should be mentioned here. The former is intended to investigate the possibility of a rule differentiating a fact, which s itself subject to another rule. Rule hierarchies could be used to intercept situations when no specific rule of one level is applicable. In this way, the aim is to describe the behaviour, respectively the target behaviour, of the automaton more and more precisely and completely.

\section{Acknowledgement}
 The research leading to these results has been funded by the German Federal Ministry for Economic Affairs and Climate Action within the project ``Verifikations- und Validierungsmethoden automatisierter Fahrzeuge Level 4 und 5''. The authors would like to thank the consortium for the successful cooperation.
In addition, the authors would like to thank the partners of the project for the successful cooperation, constructive feedback and discussions during the development of the PSM.

\pagebreak

\printbibliography

\end{document}